\documentclass[12pt]{elsarticle}
\usepackage{amsmath, amssymb, amsthm, mathtools, bbm, hyperref}
\usepackage{graphicx, xcolor}
\usepackage{orcidlink}
\usepackage[paper=a4paper,left=2.5cm,right=2.5cm,top=2.5cm,bottom=2.5cm]{geometry}

\makeatletter
\def\ps@pprintTitle{%
   \let\@oddhead\@empty
   \let\@evenhead\@empty
   \def\@oddfoot{\reset@font\hfil\thepage\hfil}
   \let\@evenfoot\@oddfoot
}
\makeatother

\makeatletter
\def\blfootnote{\gdef\@thefnmark{}\@footnotetext}
\makeatother

\newtheorem{theorem}{Theorem}[section]
\newtheorem{proposition}[theorem]{Proposition}

\theoremstyle{definition}
\newtheorem{definition}[theorem]{Definition}

\theoremstyle{remark}

\newcommand*{\ldblbrace}{\{\mskip-6mu\{}
\newcommand*{\rdblbrace}{\}\mskip-6mu\}}
\newcommand{\F}{\mathbb{F}}
\newcommand{\Fq}{\mathbb{F}_{q}}

\newcommand{\Z}{\mathbb{Z}}
\newcommand{\calG}{\mathcal{G}}
\newcommand{\calM}{\mathcal{M}}
\newcommand{\calP}{\mathcal{P}}
\newcommand{\mspan}{\operatorname{mspan}}
\newcommand{\rank}{\operatorname{rank}}
\newcommand{\height}{\operatorname{ht}}

\newcommand{ \bs }[1]{ \boldsymbol{#1} }

\numberwithin{equation}{section}

\begin{document}

\begin{frontmatter}

\title{Vector Multispaces and Multispace Codes\tnoteref{t1}}

\tnotetext[t1]{This work was supported by the Secretariat for Higher Education
and Scientific Research of the Autonomous Province of Vojvodina (project number
142-451-2686/2021).}

\author[ftn]{Mladen~Kova\v{c}evi\'{c} \orcidlink{0000-0002-2395-7628}}
%\cortext[cor]{kmladen@uns.ac.rs}

\address[ftn]{Faculty of Technical Sciences, 
              University of Novi Sad, Serbia\\
							email: kmladen@uns.ac.rs}

%\subjclass[2020]{15A03, 94B25, 05E16, 06B99}

%\date{today}

\begin{abstract}
Basic algebraic and combinatorial properties of finite vector spaces in which
individual vectors are allowed to have multiplicities larger than $ 1 $ are
derived.
An application in coding theory is illustrated by showing that multispace
codes that are introduced here may be used in random linear network coding
scenarios, and that they generalize standard subspace codes (defined
in the set of all subspaces of $ \F_q^n $) and extend them to an infinitely
larger set of parameters.
In particular, in contrast to subspace codes, multispace codes of arbitrarily
large cardinality and minimum distance exist for any fixed $ n $ and $ q $.
\end{abstract}


\end{frontmatter}

\blfootnote{\textup{2020} \textit{Mathematics Subject Classification}:
Primary: 15A03, 94B25. Secondary: 05E16, 06B99.}

\section{Introduction}
\label{sec:intro}

The aim of this note is to state basic properties of vector spaces in which
individual vectors are allowed to have multiplicities larger than $ 1 $.
Despite being an obvious linear analog of a multiset, this notion does not
seem to have been thoroughly studied in linear algebra%
\footnote{With several exceptions, including the recent works \cite{krotov, krotov2}
on certain covering problems.} \cite{roman}.
Our main motivation for analyzing it is the fact that it offers the possibility
of extending the so-called subspace codes \cite{koetter+kschischang, wang},
defined in the set of all subspaces of a given finite vector space, to a much
larger (in fact infinitely larger) set of parameters.

%In Section~\ref{sec:multispaces} we define multispaces and state some of
%their basic algebraic and combinatorial properties, while
%in Section~\ref{sec:codes} we discuss in more detail the mentioned application
%in coding theory.

\subsection*{Terminology and notation}

A multiset $ M $ over a universe $ U $ is a pair $ (U, \mu_M) $, where
$ \mu_M \colon U \to \{0, 1, 2, \ldots\} $ is a multiplicity function,
i.e., $ \mu_M(a) $ is the multiplicity of an element $ a \in U $ in $ M $.
The underlying set of $ M $ is defined as
$ \underline{M} = \operatorname{supp} \mu_M = \{a \in U : \mu_M(a) > 0\} $.
In cases when $ \mu_M(U) \subseteq \{0, 1\} $, $ M $ is identified with its
underlying set $ \underline{M} $.
The size, or cardinality, of a multiset $ M $ is the sum of the multiplicities
of all its elements, $ |M| = \sum_{a \in U} \mu_M(a) $.
The intersection and union of two multisets $ M_1 = (U, \mu_{M_1}) $,
$ M_2 = (U, \mu_{M_2}) $, are the multisets $ M_1 \cap M_2 = (U, \min\{\mu_{M_1}, \mu_{M_2}\}) $
and $ M_1 \cup M_2 = (U, \max\{\mu_{M_1}, \mu_{M_2}\}) $, respectively.
We say that $ M_1 $ is contained in $ M_2 $, written $ M_1 \subseteq M_2 $,
if $ \mu_{M_1} \leqslant \mu_{M_2} $.

A multiset can also be specified (somewhat informally) by explicitly listing
all its elements.
For example, we think of $ M = \ldblbrace a, a, b, b, b \rdblbrace $, $ a \neq b $,
also written as $ M = \ldblbrace a^2, b^3 \rdblbrace $ when there is no risk of
confusing superscripts with exponents, as a multiset with $ \underline{M} = \{a, b\} $,
$ \mu_M(a) = 2, \mu_M(b) = 3 $, and $ |M| = 5 $.
Given another multiset $ M' = \ldblbrace a^3, b^1, c^2 \rdblbrace $,
we have $ M \cap M' = \ldblbrace a^2, b^1 \rdblbrace $ and
$ M \cup M' = \ldblbrace a^3, b^3, c^2 \rdblbrace $.

Throughout the paper, $ q $ is assumed to be a prime power, $ \F_q $ denotes the
field with $ q $ elements, and $ \F^n_q $ denotes the $ n $-dimensional vector
space over $ \F_q $.
The set of all subspaces of $ \F_q^n $ is denoted by $ \calP_q(n) $ (the so-called
projective space), and the set of all $ k $-dimensional subspaces of $ \F_q^n $ by
$ \calG_q(n, k) $ (the so-called Grassmannian).
The cardinality of $ \calG_q(n, k) $ is expressed through the Gaussian, or
$ q $-binomial, coefficients \cite[Section 24]{vanlint+wilson}:
\begin{equation}
\label{eq:grassmannian}
  \left| \calG_q(n, k) \right|
	= \binom{n}{k}_{\! q}
	= \frac{(q^n-1)\cdots(q^{n-k+1}-1)}{(q^k-1)\cdots(q-1)} ,
\end{equation}
where it is understood that $ \binom{n}{0}_{\!q} = 1 $, and
$ \binom{n}{k}_{\!q} = 0 $ when $ k < 0 $ or $ k > n $.

The $ \F_q $-linear span of a set $ B \subseteq \F_q^n $, which is the set
of all vectors in $ \F_q^n $ that can be obtained as a linear combination
of those in $ B $ (with coefficients from $ \F_q $), is denoted by
$ \operatorname{span}_q(B) $.

\section{Vector Multispaces}
\label{sec:multispaces}

The following definition appears also in \cite{krotov} (in a slightly different form).

\begin{definition}
\label{def:mspan}
Let $ B = \ldblbrace b_1, \ldots, b_m \rdblbrace $ be a multiset over
$ \F_q^n $ ($ b_i $'s are not necessarily all distinct).
The $ \F_q $-linear \emph{multispan} of $ B $, denoted $ \mspan_q(B) $,
is defined as the multiset over $ \Fq^n $ in which the multiplicity of
an element $ v $ is equal to the number of different $ m $-tuples
$ (\alpha_{1}, \ldots, \alpha_{m}) \in \F_q^m $ satisfying
$ \sum_{i=1}^m \alpha_{i} b_i = v $.
If $ B $ is empty, we define $ \mspan_q(B) = \{0\} $.
\end{definition}

It follows from the definition that
\begin{equation}
\label{eq:mspansize}
  \left| \mspan_{q}(B) \right| = q^{|B|} .
\end{equation}

\begin{definition}
\label{def:mspace}
A multiset $ W $ over $ \F_q^n $ is said to be a \emph{multispace} over
$ \F_q^n $ if there exists a multiset $ B $ such that $ W = \mspan_q(B) $.
We then say that $ W $ is generated by $ B $ and define the rank of $ W $
as $ \rank(W) = |B| = \log_q\!|W| $ and the dimension of $ W $ as
$ \dim(W) = \dim(\operatorname{span}_q(B)) $.

We denote the collection of all rank-$ m $ multispaces over $ \F_q^n $
by $ \calM_q(n, m) $, and the collection of all multispaces over $ \F_q^n $
by $ \calM_q(n) = \bigcup_{m=0}^\infty \calM_q(n, m) $.
%
%Additionally, we identify all multisets $ W $ such that $ \underline{W} = \{0\} $
%with $ \{0\} $ and consider this a multispace of rank $ 0 $ (and dimension
%$ 0 $).
%{\color{red} Mislim da ipak $ \{0^{q^t}\} $ treba posmatrati kao multiprostor ranga $ t $.}
\end{definition}

A multispace $ W $ over $ \F_q^n $ is a subspace of $ \F_q^n $, i.e., the
multiplicity of each element in $ W $ is $ 1 $, if and only if it is generated
by a set of linearly independent vectors.
In that case, and only in that case, we have $ \rank(W) = \dim(W) $.
The following claim, which states that a multispace is uniquely determined
by its underlying space and its rank, is a generalization of this fact.

\begin{proposition}
\label{thm:multiplicity}
A multiset $ W $ over $ \F_q^n $ is a multispace if and only if the following
two conditions hold:
(i) $ \underline{W} $ is a subspace of $ \F_q^n $,
(ii) every element of $ W $ has multiplicity $ q^{t} $, for some nonnegative
integer $ t $.
Furthermore, $ t = \rank(W) - \dim(W) $.
\end{proposition}
\begin{proof}
Let $ W $ be a multispace over $ \Fq^n $.
Condition \emph{(i)} follows directly from Definitions \ref{def:mspan}
and \ref{def:mspace}: if $ W $ is generated by $ B $, then
$ \underline{W} = \operatorname{span}_q(B) \in \calP_q(n) $.
To establish condition \emph{(ii)} it suffices to show that every element
in $ W $ has the same multiplicity; since $ |W| = q^{\rank(W)} $ and
$ |\underline{W}| = q^{\dim(W)} $, it will follow that this multiplicity
is $ |W| / |\underline{W}| = q^{\rank(W) - \dim(W)} $.
Suppose that $ W $ is generated by
$ B = \ldblbrace b_1, \ldots, b_m \rdblbrace $, $ b_i \in \F_q^n $, where
$ m = \rank(W) $.
Let $ \alpha = ( \alpha_{1}, \ldots, \alpha_{m} ) \in \Fq^m $ be an $ m $-tuple
of coefficients that produces a vector $ v \in \underline{W} $, i.e.,
$ \sum_{i=1}^m \alpha_{i} b_i = v $.
Then it is obvious that the $ m $-tuple $ \alpha + \beta $, where
$ \beta \in \Fq^m $, produces $ v $ as well if and only if the $ m $-tuple
$ \beta $ produces $ 0 $.
This implies that, if the vector $ 0 $ has multiplicity $ s $ in $ W $,
then so does every other vector from $ \underline{W} $.

For the converse statement, let $ B' = \{b_1, \ldots, b_{\dim(\underline{W})}\} $
be a basis of $ \underline{W} $, and consider the multiset
$ B = \ldblbrace 0^{t}, b_1, \ldots, b_{\dim(\underline{W})} \rdblbrace $.
%$ B = \ldblbrace \underbrace{0, \ldots, 0}_{t}, b_1, \ldots, b_{\dim(\underline{W})} \rdblbrace $.
Then $ \mspan_q(B) $ is a multispace whose underlying space is $ \underline{W} $
and in which the multiplicity of each element is $ q^t $, which is precisely $ W $.
Hence, $ W $ is a multispace of rank $ \dim(\underline{W}) + t $, as claimed.
\end{proof}

Note that the generating set for $ W $ could have also been obtained from the
basis of $ \underline{W} $ by adding to this basis any $ \rank(W) - \dim(W) $
vectors from $ \underline{W} $, not necessarily $ \rank(W) - \dim(W) $ copies
of the $ 0 $ vector (different multisets can generate the same multispace).
A simple class of transformations of the generating multiset that preserve the
resulting multispace is the full-rank linear transformations.

\begin{proposition}
\label{thm:fullrank}
  Let $ \bs{b} = (b_1, \ldots, b_m) $, $ b_i \in \Fq^n $.
Let $ T $ be an $ m \times m $ full-rank matrix over $ \Fq $, and let
$ \bs{b'} = \bs{b} T $.
Let $ B $ denote the multiset of elements of $ \bs{b} $, i.e.,
$ B = \ldblbrace b_1, \ldots, b_m \rdblbrace $, and likewise
$ B' = \ldblbrace b'_1, \ldots, b'_m \rdblbrace $.
Then $ \mspan_q(B) = \mspan_q(B') $.
\end{proposition}
\begin{proof}
%If $ \bs{b} = (0, \ldots, 0) $, then $ \bs{b'} = (0, \ldots, 0) $,
%and hence $ \mspan_q(B) = \mspan_q(B') = \{0\} $.
Clearly, $ \rank(\mspan_q(B)) = \rank(\mspan_q(B')) = m $.
Furthermore, since $ \bs{b'} $ is obtained from $ \bs{b} $ by a full-rank
linear transformation, we necessarily have
$ \operatorname{span}_q(B) = \operatorname{span}_q(B') $.
This means that the multispaces $ \mspan_q(B) $ and $ \mspan_q(B') $ have
the same rank and the same underlying space, and so, by Proposition~\ref{thm:multiplicity},
they are equal.
\end{proof}

The following proposition gives an algebraic characterization of multispaces.
Let $ \F_{q^\ell} $ be an extension field of $ \Fq $.
A polynomial of the form
\begin{equation}
\label{eq:qpolynomial}
  L(x) = \sum_{i=0}^d \alpha_i x^{q^i}
\end{equation}
with $ \alpha_i \in \F_{q^\ell} $ is called a linearized polynomial, or
$ q $-polynomial, over $ \F_{q^\ell} $ \cite[Chapter 3]{fields}.
Such a polynomial is a linear operator on any extension field $ \F $ of
$ \F_{q^\ell} $ (here extension field is thought of as a vector space
over the base field).

\begin{proposition}
\label{thm:qpolynomial}
  $ W $ is a multispace over $ \Fq^n $ if and only if it is the multiset
of all roots of a linearized polynomial over a subfield $ \F_{q^\ell} $ of
$ \F_{q^n} $, for some $ \ell | n $.
\end{proposition}
\begin{proof}
Let $ L(x) $ be a linearized polynomial over $ \F_{q^\ell} $, and suppose
that an extension field $ \F_{q^n} $ of $ \F_{q^\ell} $ contains all the
roots of $ L(x) $.
Then, by \cite[Theorem 3.50]{fields}, every root has the same multiplicity,
which is of the form $ q^t $, and the roots form a vector subspace of $ \F_{q^n} $
(where $ \F_{q^n} $ is regarded as a vector space over $ \Fq $).
By Proposition \ref{thm:multiplicity}, this means that the multiset of all
roots of $ L(x) $ is a multispace over $ \Fq^n $.

Conversely, let $ W $ be a multispace over $ \Fq^n $, and denote $ t = \rank{W}-\dim{W} $.
Then
\begin{equation}
  L(x) = \prod_{v \in W} (x - v) = \prod_{v \in \underline{W}} (x - v)^{q^{t}}
\end{equation}
is a linearized polynomial over $ \F_{q^n} $ (where, again, the vectors in
$ \Fq^n $ are identified with elements of $ \F_{q^n} $) \cite[Theorem 3.52]{fields}.
It is easy to see that the multiset of roots of $ L(x) $ is precisely $ W $.
\end{proof}

%Let us denote the collection of all rank-$ m $ multispaces over $ \F_q^n $
%by $ \calM_q(n; m) $, and the collection of all multispaces over $ \F_q^n $
%by $ \calM_q(n) = \bigcup_{m=0}^\infty \calM_q(n; m) $.

We next consider the properties of the partially ordered set $ (\calM_q(n), \subseteq) $.
Recall \cite{birkhoff} that a poset $ (P, \preceq) $ is said to be a lattice
if every two elements $ x, y \in P $ have a greatest lower bound, denoted
$ x \land y $ (that is, $ x \land y \preceq x $, $ x \land y \preceq y $,
and any other element $ w $ satisfying $ w \preceq x $, $ w \preceq y $,
also satisfies $ w \preceq x \land y $), and a least upper bound, denoted
$ x \lor y $.
The lattice is graded if there exists a function $ r \colon P \to \Z $ such
that, if $ y $ covers $ x $ (meaning that $ x \preceq y $ and there exists
no $ z $ satisfying $ x \prec z \prec y $), then $ r(y) = r(x) + 1 $.
A lattice is modular if $ x \prec y $ implies $ x \lor (z \land y) = (x \lor z) \land y $.

\begin{proposition}
\label{thm:lattice}
  The poset $ (\calM_q(n), \subseteq) $ is a graded modular lattice.
\end{proposition}
\begin{proof}
Consider two multispaces $ W_1, W_2 \in \calM_q(n) $.
By definition of multiset intersection, $ W_1 \cap W_2 $ is a multiset
whose underlying set is
$ \underline{W_1 \cap W_2} = \underline{W_1} \cap \underline{W_2} $ and
in which the multiplicity of each element from $ \underline{W_1 \cap W_2} $
is $ \min\!\left\{ q^{\rank(W_1) - \dim(W_1)} , \, q^{\rank(W_2) - \dim(W_2)} \right\} $.
Since the underlying set is in fact a subspace of $ \Fq^n $, and each
element has multiplicity of the form $ q^t $, it follows from
Proposition~\ref{thm:multiplicity} that $ W_1 \cap W_2 $ is a multispace of
rank
\begin{equation}
\label{eq:meetrank}
\begin{aligned}
  &\rank(W_1 \cap W_2)  \\
	 &= \dim\!\left(\underline{W_1} \cap \underline{W_2} \right)
     + \min\!\big\{ \rank(W_1) - \dim(W_1), \,  \rank(W_2) - \dim(W_2)\big\} .
\end{aligned}
\end{equation}
Furthermore, it is easy to see that any multispace that is contained
in both $ W_1 $ and $ W_2 $ is also contained in $ W_1 \cap W_2 $.
Consequently, the greatest lower bound on $ W_1 $ and $ W_2 $ in the
poset $ (\calM_q(n), \subseteq) $, denoted $ W_1 \land W_2 $, exists
and is given by $ W_1 \cap W_2 $.

The least upper bound on $ W_1 $ and $ W_2 $ in the poset
$ (\calM_q(n), \subseteq) $ exists as well -- it is the multispace,
denoted $ W_1 \lor W_2 $, with the underlying space
$ \underline{W_1 \lor W_2} = \underline{W_1} + \underline{W_2} $
and rank
\begin{equation}
\label{eq:joinrank}
\begin{aligned}
	&\rank(W_1 \lor W_2)  \\
	 &= \dim\!\left(\underline{W_1} + \underline{W_2}\right)
     + \max\!\big\{ \rank(W_1) - \dim(W_1), \,  \rank(W_2) - \dim(W_2)\big\} .
\end{aligned}
\end{equation}
To see this, note that (1) if a multispace contains both $ W_1 $ and
$ W_2 $, then its underlying space must contain both $ \underline{W_1} $
and $ \underline{W_2} $, and hence also $ \underline{W_1} + \underline{W_2} $,
and (2) the multiplicity of elements in any multispace containing both
$ W_1 $ and $ W_2 $ is at least
$ \max\!\big\{ q^{\rank(W_1) - \dim(W_1)} , \\ q^{\rank(W_2) - \dim(W_2)} \big\} $.

We next show that the lattice $ (\calM_q(n), \subseteq) $ is graded, the
rank function being precisely $ \rank(\cdot) $ from Definition~\ref{def:mspace}.
Suppose that $ W_2 $ covers $ W_1 $, meaning that $ W_1 \subseteq W_2 $
and there exists no $ W_3 \in \calM_q(n) $ such that $ W_1 \subsetneq W_3 \subsetneq W_2 $,
and recall that the multiplicities of the elements in $ W_i $ are
$ q^{\rank(W_i) - \dim(W_i)} $, $ i = 1,2 $ (Proposition~\ref{thm:multiplicity}).
Then either $ \underline{W_2} = \underline{W_1} $, or $ \underline{W_2} $
covers $ \underline{W_1} $ in the lattice $ (\calP_q(n), \subseteq) $.
In the first case we must have $ q^{\rank(W_2) - \dim(W_2)} = q^{\rank(W_1) - \dim(W_1) + 1} $,
and hence $ \rank(W_2) = \rank(W_1) + 1 $ (since $ \dim(W_2) = \dim(W_1) $
by assumption).
In the second case we must have $ q^{\rank(W_2) - \dim(W_2)} = q^{\rank(W_1) - \dim(W_1)} $,
and again $ \rank(W_2) = \rank(W_1) + 1 $ (since $ \dim(W_2) = \dim(W_1) + 1 $
by assumption).

One can see from \eqref{eq:meetrank}, \eqref{eq:joinrank}, and the fact
that $ \dim(V_1 \cap V_2) + \dim(V_1 + V_2) = \dim(V_1) + \dim(V_2) $ for
arbitrary vector spaces $ V_1, V_2 $, that the following relation holds
for any $ W_1, W_2, \in \calM_q(n) $:
\begin{equation}
\label{eq:rankval}
  \rank(W_1 \land W_2) + \rank(W_1 \lor W_2) = \rank(W_1) + \rank(W_2) .
\end{equation}
This implies that $ (\calM_q(n), \subseteq) $ is a metric lattice
\cite[Section X.2]{birkhoff}, and hence also a modular lattice
\cite[Chapter X, Theorem 2]{birkhoff}.
The proof is complete.
\end{proof}

The well-known linear lattice $ (\calP_q(n), \subseteq) $ is a finite
sublattice of $ (\calM_q(n), \subseteq) $;
see Figure~\ref{fig:lattice} for an illustration.

\begin{figure}%[h]
 \centering
  \includegraphics[width=\columnwidth]{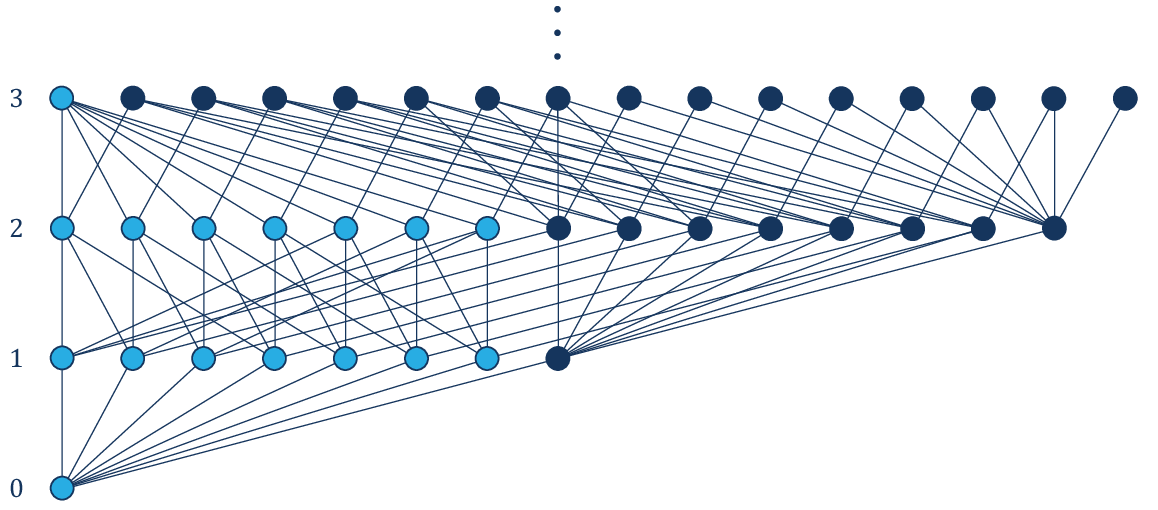}
\caption{Hasse diagram of the set of all multispaces over $ \F_2^3 $ of
rank up to $ 3 $, partially ordered by inclusion. Subspaces of $ \F_2^3 $
are represented by light-blue nodes.}
\label{fig:lattice}
\end{figure}%

\begin{proposition}
\label{thm:cardinality}
The number of multispaces of rank $ m $ over $ \F_q^n $ is
given by
\begin{equation}
\label{eq:Mqnm}
  \left| \calM_q(n, m) \right| = \sum_{k=0}^m \binom{n}{k}_{\!q} .
\end{equation}
The number of multispaces covered by $ W $, $ \rank(W) \geqslant 1 $,
in $ (\calM_q(n), \subseteq) $ is
%\footnote{$ \mathbbm{1}(\cdot) $ is the indicator function having value $ 1 $
%or $ 0 $ depending on whether or not the clause in the parenthesis is true.}
\begin{equation}
\label{eq:covered}
  %\mathbbm{1}(\rank(W) > \dim(W)) + \sum_{i=0}^{\dim(W) - 1} q^i .
	\begin{cases}
    \sum\limits_{i=0}^{\dim(W) - 1} q^i ,  & \quad \operatorname{rank}(W) = \dim(W) \\
    1 + \sum\limits_{i=0}^{\dim(W) - 1} q^i , & \quad \operatorname{rank}(W) > \dim(W)
  \end{cases}.
\end{equation}
The number of multispaces that cover $ W $ in $ (\calM_q(n), \subseteq) $ is
\begin{equation}
\label{eq:covering}
  1 + \sum_{i=0}^{n - \dim(W) - 1} q^i .
\end{equation}
\end{proposition}
\begin{proof}
%By Proposition~\ref{thm:multiplicity}, a multispace is uniquely determined
%by its underlying space and its rank.
Every subspace of $ \F_q^n $ of dimension $ 0, \ldots, m $ is the underlying
space of a unique multispace of rank $ m $.
This fact, together with \eqref{eq:grassmannian}, implies \eqref{eq:Mqnm}.

Let us show \eqref{eq:covered}.
If $ \rank(W) = \dim(W) $, then $ W $ is a subspace of $ \Fq^n $ and hence
the multispaces it covers are the $ (\dim(W)-1) $-subspaces of $ W $, of
which there are
$ \binom{\dim(W)}{\dim(W)-1}_{\! q} = \binom{\dim(W)}{1}_{\! q} = \frac{q^{\dim(W)}-1}{q-1} $
(see~\eqref{eq:grassmannian}).
If $ \rank(W) > \dim(W) $, then there are $ \binom{\dim(W)}{1}_{\! q} + 1 $
multispaces of rank $ \rank(W) - 1 $ covered by $ W $:
$ \binom{\dim(W)}{1}_{\! q} $ multispaces having as the underlying space a
$ (\dim(W)-1) $-subspace of $ \underline{W} $ and in which the multiplicity
of each element is $ q^{\rank(W) - \dim(W)} $, and $ 1 $ multispace with
$ \underline{W} $ as the underlying space and with the multiplicity of each
element being $ q^{\rank(W) - \dim(W) - 1} $.

\eqref{eq:covering} is shown in a similar way.
There are $ \binom{n-\dim(W)}{1}_{\! q} + 1 $ multispaces of rank $ \rank(W) + 1 $
covering $ W $:
$ \binom{n-\dim(W)}{1}_{\! q} $ multispaces having as the underlying space
a $ (\dim(W)+1) $-superspace of $ \underline{W} $ and in which the multiplicity
of each element is $ q^{\rank(W) - \dim(W)} $, and $ 1 $ multispace having
$ \underline{W} $ as the underlying space and with the multiplicity of each
element being $ q^{\rank(W) - \dim(W) + 1} $.
\end{proof}

Note that $ \left| \calM_q(n, m) \right| $ does not depend on $ m $ when
$ m \geqslant n $, i.e.,
\begin{equation}
  \left| \calM_q(n, m) \right|
	 = \left| \calM_q(n, n) \right|
	 = \left| \calP_q(n) \right| ,  \quad
	 \forall m \geqslant n . 
\end{equation}
In fact, every segment of the lattice $ (\calM_q(n), \subseteq) $ involving
multispaces of rank $ m $ and $ m + 1 $, where $m > n$, namely
$ (\calM_q(n, m) \cup \calM_q(n, m+1), \subseteq) $, is an identical copy
of the segment $ (\calM_q(n, n) \cup \calM_q(n, n+1), \subseteq) $.

A natural metric can be associated with a lattice satisfying \eqref{eq:rankval}
\cite[Chapter X]{birkhoff}, namely:
\begin{equation}
\label{eq:metric}
\begin{aligned}
  d(W_1, W_2)
   &= \rank(W_1 \lor W_2) - \rank(W_1 \land W_2)  \\
	 &= 2 \rank(W_1 \lor W_2) - \rank(W_1) - \rank(W_2)  \\
	 &= \rank(W_1) + \rank(W_2) - 2 \rank(W_1 \land W_2) .
\end{aligned}
\end{equation}
This distance can also be interpreted as the length of the shortest path
from $ W_1 $ to $ W_2 $ in the Hasse diagram of $ (\calM_q(n), \subseteq) $.

We next examine how the distance between two multispaces relates to the
distance between their underlying spaces.
Define the height of a multispace by $ \height(W) = \rank(W) - \dim(W) $.
(Recall from Proposition \ref{thm:multiplicity} that the multiplicity of
each element of $ W $ is $ q^{\height(W)} $.)
Then \eqref{eq:meetrank} can be written as
\begin{equation}
\label{eq:meetrank2}
  \rank(W_1 \cap W_2)
	 =\dim\!\left(\underline{W_1}\cap\underline{W_2}\right) + \min\!\big\{ \height(W_1) , \,  \height(W_2) \big\} ,
\end{equation}
which implies
\begin{equation}
\label{eq:d}
\begin{aligned}
  d(W_1, W_2)
	 &= \rank(W_1) + \rank(W_2) - 2\rank(W_1 \cap W_2)  \\
   &= \dim\!\left(\underline{W_1}\right) + \height(W_1) + \dim\!\left(\underline{W_2}\right) + \height(W_2)  \\
      &\phantom{=XX} - 2 \dim\!\left(\underline{W_1}\cap\underline{W_2}\right) - 2 \min\!\big\{ \height(W_1) , \,  \height(W_2) \big\}  \\
   &= d\!\left(\underline{W_1}, \underline{W_2}\right) + \left| \height(W_1) - \height(W_2) \right| ,
\end{aligned}
\end{equation}
where in the last equality we used the fact that $ a + b - 2\min\{a , b\} = |a - b| $.
This relation directly implies the following claim.

\begin{proposition}
\label{thm:distance}
  For any $ W_1, W_2 \in \calM_q(n) $,
\begin{equation}
  d(W_1, W_2) \geqslant d\!\left(\underline{W_1}, \underline{W_2}\right) ,
\end{equation}
with equality if and only if $ \height(W_1) = \height(W_2) $.

For any $ W_1, W_2 \in \calM_q(n, m) $, \begin{equation}
\label{eq:metric1}
  d(W_1, W_2) = d(\underline{W_1}, \underline{W_2}) + |\dim(W_1) - \dim(W_2)|
\end{equation}
and, consequently,
\begin{equation}
\label{eq:metric2}
  d(\underline{W_1}, \underline{W_2}) \leqslant d(W_1, W_2) \leqslant 2 d(\underline{W_1}, \underline{W_2}) ,
\end{equation}
where equality on the left-hand side holds if and only if $ \dim(W_1) = \dim(W_2) $,
and the one on the right-hand side holds if and only if
$ \underline{W_1} \subseteq \underline{W_2} $ or $ \underline{W_2} \subseteq \underline{W_1} $.
\end{proposition}

The graph $ \Gamma_{q}(n, m) $ associated with the metric space $ (\calM_q(n, m), d) $
is obtained by joining by an edge any two nodes/multispaces that are at
distance $ 2 $ (the distance between multispaces of the same rank is
necessarily even, see \eqref{eq:metric1}).
The graph $ \Gamma_{q}(n,m) $ contains as subgraphs the Grassmann graphs
corresponding to $ (\calG_q(n,k), d) $, $ k = 1, \ldots, m $.
Unlike Grassmann graphs, however, $ \Gamma_{q}(n, m) $ is not distance-regular.

%Kontraprimer: $ d(\{001^2\}, \{010^2\}) = 2 $, $ d(\{001, 010\}, \{011, 100\}) = 2 $.
%Broj multiprostora na rastojanju 1 od $ \{001^2\} $ i rastojanju 1 od $ \{010^2\} $ je 1 ($ \{001, 010\} $).
%Broj multiprostora na rastojanju 1 od $ \{001, 010\} $ i rastojanju 1 od $ \{011, 100\} $ je 2 ($ \{001, 100\} $, $ \{010, 100\} $).

\section{Multispace Codes}
\label{sec:codes}

In this section we illustrate an application of multispaces in coding theory.
The main purpose of this discussion is to present the framework for
\emph{multispace codes} and their applications, leaving a more in depth
investigation of such codes for future work.

Subspace codes are codes in $ \calP_q(n) $; when restricted to $ \calG_q(n, k) $,
they are called constant-dimension subspace codes.
These objects were proposed in \cite{koetter+kschischang}%
\footnote{See also the earlier work \cite{wang} for a different application
of codes in $ \calP_q(n) $.}
as appropriate constructs for error correction in communication channels that
deliver to the receiver random linear combinations of the transmitted vectors
from $ \Fq^n $.
A practical instantiation of such a channel would be a network whose nodes/routers
forward on their outgoing links random linear combinations of the incoming
vectors/packets, a paradigm known as `random linear network coding' \cite{ho}.
Given a channel of this kind, one naturally comes to the idea that information
the transmitter is sending should be encoded in an object that is invariant
(with high probability) under random linear transformations, namely a vector
subspace of $ \Fq^n $.
A code for such a channel should therefore be defined in $ \calP_q(n) $ and
should, as always, satisfy certain minimum distance requirements in order for
the receiver to be able to correct a certain number of errors.
The appropriate way of measuring the distance between subspaces depends on the
details of the model, but the most common choice is precisely the metric \eqref{eq:metric}.
In the described setting, the transmitter would typically send a basis of the
selected codeword/subspace through the channel, and the receiver's task would
be to reconstruct the codeword from the received vectors.

We propose the use of codes in $ \calM_q(n) $ for the same purpose.
The scenario we have in mind is the following: the transmitter sends a
generating multiset $ B $ of the selected codeword/multispace $ W $ (over
$ \Fq^n $), and the channel delivers to the receiver a multiset $ B' $ of
random linear combinations of the transmitted vectors.
Note that multiple copies of the same packet/vector can easily be generated,
so the assumption that the transmitter sends a multiset of vectors is physically
well-justified and quite natural in this context.
The scenario just described was in fact the motivation for defining multispaces
as in Definition~\ref{def:mspace}.
It follows from Proposition~\ref{thm:fullrank} that, if the channel action
is represented by a full-rank $ \rank(W) \times \rank(W) $ matrix, then the
transmitted multispace is preserved in the channel, i.e., the receiver will
reconstruct it perfectly from the received multiset $ B' $.
In general, however, the channel may introduce errors.
We next show that two reasonable types of errors in this context -- deletions
and dimension-reductions -- result in the received multispace that is at a
bounded distance, depending only on the number of errors, from the transmitted
multispace.
Consequently, one can design a code in $ \calM_q(n) $ that is guaranteed to
correct a given number of errors by specifying its minimum distance with
respect to the metric \eqref{eq:metric}.

\newpage
\begin{proposition}
Let $ \bs{b'} = \bs{b} T $, where $ \bs{b} = (b_1, \ldots, b_m) $, $ \bs{b'} = (b'_1, \ldots, b'_{m'}) $,
$ b_i, b'_i \in \Fq^n $, and $ T $ is an $ m \times m' $ matrix over $ \Fq $.
Denote also $ B = \ldblbrace b_1, \ldots, b_m \rdblbrace $,
$ B' = \ldblbrace b'_1, \ldots, b'_m \rdblbrace $.\newline
(1) If $ m' = m-s $ and $ \rank T = m' $, then $ d(\mspan_q(B), \mspan_q(B')) = s $.\newline
(2) If $ m' = m $ and $ \rank T = m - s $, then $ d(\mspan_q(B), \mspan_q(B')) \leqslant 2s $.
\end{proposition}
\begin{proof}
We think of $ B $, $ B' $ as the transmitted and the received multisets
of vectors, respectively, and of $ W = \mspan_q(B) $, $ W' = \mspan_q(B') $
as the corresponding multispaces.

\emph{(1)}
The channel action (i.e., the action of the operator $ T $) can in this
case be thought of as first transforming $ \bs{b} $ to $ \bs{\tilde{b}} $
by a full-rank $ m \times m $ matrix, and then transforming $ \bs{\tilde{b}} $
to $ \bs{b'} $ by deleting $ s $ of the former's coordinates.
We know from Proposition~\ref{thm:fullrank} that $ W = \tilde{W} $.
Also, it is straightforward to see that deleting $ s $ elements of the generating
multiset $ \tilde{B} $, i.e., $ B' \subseteq \tilde{B} $, $ |B'| = |\tilde{B}| - s $,
results in a multispace $ W' $ which is contained in $ \tilde{W} $ and
whose rank is $ \rank(W') = \rank(\tilde{W}) - s $.
Therefore, under the stated assumptions, $ W \land W' = W' $, and hence
$ d(W, W') = s $.

\emph{(2)}
This scenario corresponds to the case when the channel delivers $ m $
vectors to the receiver (as many as were transmitted), but the transformation
is not full-rank.
A simple linear algebra shows that, under these assumptions, $ \dim(W') \geqslant \dim(W) - s $.
Since also $ W' \subseteq W $ and $ \rank(W) = \rank(W') = m $, we conclude
from Proposition~\ref{thm:distance} that $ d(W, W') \leqslant 2s $.
\end{proof}

The main advantage of the transition from subspaces to multispaces that is
proposed here is that the code space is in the latter case infinitely larger
-- $ \calM_q(n) $ instead of $ \calP_q(n) $.
In particular, if the parameters $ n $ and $ q $ (which, in the network
coding scenario, correspond to packet length and alphabet size, respectively)
are fixed, subspace codes are necessarily bounded, both in cardinality and
minimum distance.
This fact stands in sharp contrast to most other communication scenarios in
which codes of arbitrarily large cardinality and minimum distance exist even
over a fixed alphabet.
%, and one can meaningfully consider various asymptotic regimes.
With multispaces on the other hand, taking, say, $ \bigcup_{j=0}^{m} \calM_q(n, j) $
as the code space, one can consider the asymptotic regime $ m \to \infty $
even if $ n $ and $ q $ are fixed.
In this regime the size of the code space scales as
\begin{equation}
  \sum_{j=0}^{m} |\calM_q(n, j)|  \;\sim\;  m \cdot \sum_{k=0}^{n} \binom{n}{k}_{\! q}
\end{equation}
(see \eqref{eq:Mqnm}) and is unbounded.

\vspace{1cm}
To conclude the paper, we again emphasize that the aim of the brief discussion
in this section was to introduce the notion of multispace codes -- a generalization
of subspace codes -- and argue that they can be used in the same scenarios as
subspace codes, with an important advantage that they are capable of achieving
larger (in some cases infinitely larger) information rates.
Further work on the subject would entail devising concrete constructions of
codes in $ \bigcup_{j=0}^{m} \calM_q(n, j) $ having a desired minimum distance,
and deriving bounds on the cardinality of optimal codes (similar to those in,
e.g., \cite{etzion+vardy}).
%The latter will likely be a more challenging task than its subspace counterpart
%(see, e.g., \cite{etzion+vardy}) because the graph $ \Gamma_q(n, m) $ is not
%distance-regular.
We note that certain bounds on codes in modular lattices, namely sphere-packing
and Singleton bounds, have already been studied in \cite{kendziora+schmidt}.

\newpage
%\vspace{5mm}
\bibliographystyle{elsarticle-num}

\end{document}